\documentclass[smallabstract,smallcaptions]{dccpaper}

\usepackage{epsfig}
\usepackage{cite}
\usepackage{amsmath}
\usepackage{amssymb}
\usepackage{amsthm}
\newtheorem{theorem}{Theorem}

\usepackage{color}
\usepackage{url}
\usepackage{algorithm}
\usepackage[noend]{algpseudocode}

\usepackage{bbm}

\newtheorem{definition}[theorem]{Definition}

\newlength{\figurewidth}
\newlength{\smallfigurewidth}

\usepackage{xargs}                     \usepackage[colorinlistoftodos,prependcaption,textsize=tiny]{todonotes}

\newcommand {\bool} {\{0,1\}}
\newcommand {\ising} {\{-1,1\}}

\begin{document}

\title
{\large
\textbf{Compressed Quadratization of Higher Order Binary Optimization Problems}
}

\author{%
Avradip Mandal, Arnab Roy, Sarvagya Upadhyay and Hayato Ushijima-Mwesigwa\\[0.5em]
{\small\begin{minipage}{\linewidth}\begin{center}
\begin{tabular}{c}
Fujitsu Laboratories of America, Inc,
1240 E. Arques Ave., Sunnyvale, CA, 94085, USA \\
\{amandal, aroy, supadhyay, hayato\}@us.fujitsu.com
\end{tabular}
\end{center}\end{minipage}}
}

\maketitle
\thispagestyle{empty}

\begin{abstract}

Recent hardware advances in quantum and quantum-inspired annealers promise substantial speed up for solving NP-hard combinatorial optimization problems compared to general purpose computers. These special purpose hardware are built for solving hard instances of Quadratic Unconstrained Binary Optimization (QUBO) problems. In terms of number of variables and precision these hardware are usually resource constrained and they work either in Ising space $\ising$ or in Boolean space $\bool$. Many naturally occurring problem instances are higher order in nature. The known method to reduce the degree of a higher order optimization problem uses Rosenberg's polynomial. The method works in Boolean space by reducing degree of one term by introducing one extra variable. In this work, we prove that in Ising space the degree reduction of one term requires introduction of two variables. Our proposed method of degree reduction works directly in Ising space, as opposed to converting an Ising polynomial to Boolean space and applying previously known Rosenberg's polynomial. For sparse higher order Ising problems, this results in more compact representation of the resultant QUBO problem, which is crucial for utilizing resource constrained QUBO solvers.

\end{abstract}

\section{Introduction}
As we approach the physical limitation of Moore's law \cite{schaller1997moore}, a number of institutions have started to develop novel hardware specifically for solving combinatorial optimization problems. These include  adiabatic quantum computers \cite{johnson2011quantum}, CMOS annealers \cite{aramon2019physics,daweb,yamaoka201524,yoshimura2013spatial}, Coherent Ising Machines \cite{kielpinski2016information, mcmahon2016fully}, and as well as GPUs \cite{albash2018demonstration, inagaki2016coherent, king2019quantum} among others \cite{coffrin2019evaluating}.
These novel technologies are designed to find an assignment of the binary variables $\textbf{z} = (z_1, \cdots, z_n)$ that minimizes the following objective function:
\begin{eqnarray}\label{eq:hamiltonian}
 H(\textbf{z}) = \sum_{i<j} J_{ij}z_iz_j + \sum_i h_iz_i.
 \end{eqnarray}
 for $J_{ij}, h_i \in \mathbb{R}$. If $\textbf{z} \in \ising^n$, then this problem is referred to as the \emph{Ising Model}. It is referred to as the Quadratic Unconstrained Optimization Problem (QUBO) if $ \textbf{z} \in \{0,1\}^n$. The Ising model and QUBO are equivalent via a linear transformation of the variables. 

QUBO and the Ising model are quadratic models by definition. However, many real world optimization problems consist of multi-body interactions and are naturally modeled by higher order polynomials. These higher order polynomials have been referred to by different names in the literature, these include $k$-local Hamiltonian \cite{xia2017electronic}, psuedo-boolean optimization problem \cite{tanburn2015reducing}, Polynomial Unconstrained Binary Optimisation \cite{glover2011polynomial}, Higher Order Binary Optimization (HOBO). In particular the objective is to minimize (or maximize) the following function:
\begin{eqnarray}\label{eq:hobo}
 H(\textbf{z}) =  \sum_{i_1} J_{i_1}z_{i_1} + \sum_{i_1<i_2} J_{i_1i_2}z_{i_1}z_{i_2} + 
 \dots + \sum_{i_1 <i_2<\cdots < i_N}J_{i_1i_2\cdots i_N}\prod_{j =1}^N z_j
 \end{eqnarray}
for some $N \geq 3 $ with either $\textbf{z} \in \ising^n$ or $\{0,1\}^n$ where the coefficients $J_{i_1i_2\cdots i_k}$ are real numbers. For ease of exposition, we will refer to the domains $\ising^n$, and $\{0,1\}^n$ as \emph{Ising space}, and \emph{Boolean space} respectively and we exclusively use the variables $\textbf{s}, \textbf{x}$ and $\textbf{z}$ such that $\textbf{s} \in  \ising^n$, $\textbf{x} \in  \{0,1\}^n$, and $\textbf{z}$ when a distinction is not necessary. 
A natural way to utilize the emerging hardware for solving such problems is to first transform higher-order problem to a quadratic one, by a process that has been termed as \emph{quadratization}, then solve the quadratic problem with the given hardware. Thus, a large number of studies have focused on techniques for quadratization of HOBOs. However, to the best of our knowledge, all of these methods take place within the Boolean space. Thus, current methods for solving HOBOs in Ising space consist of first transforming them into Boolean space via the transformation $\textbf{s} = 2\textbf{x} - \mathbbm{1}$ and then applying the known techniques of quadratization. However, this approach has one major drawback, a \emph{sparse} problem in Ising space \emph{is not necessarily sparse} in Boolean space and vice versa. For example the single term $ \prod_{i=1}^Ns_i = \prod_{i=1}^N(2x_i -1)$ in Ising space consists of $2^N$ terms in Boolean space.

\paragraph{Contribution:} In this work, we develop a method of degree reduction within Ising space. To the best of our knowledge, this is the first method that does not require a transformation into Boolean space, thus resulting into a compressed quadratization technique when the polynomial in Ising space is sparse. In addition, we propose two algorithms to do quadratization of any higher order Boolean polynomial over Boolean or Ising space.

\section{How to encode quadratic equality constraints}\label{sec:twoauxIsing}
Suppose we want to minimize a quadratic function $f(x_1, x_2, y)$ over binary variables $x_1, x_2, y$, subject to the constraint $y = x_1 x_2$. However, the given hardware only supports problems modeled as a QUBO or Ising Model which are unconstrained by definition. We would like to transform the problem into one suitable for the given hardware. The equivalent unconstrained optimization problem is the minimization of $g'(x_1, x_2, y) = f(x_1, x_2, x_3) + M \cdot h'(x_1, x_2, y)$, where $M$ is a large positive constant\footnote{$M$ must be bigger than $\min f(x_1, x_2, x_1 x_2) - \min f(x_1, x_2, y)$} and $h'(x_1, x_2, y) = (y - x_1 x_2)^2$. Note that, in Ising space we have $x^2 = 1$ and in Boolean space we have $x^2 = x$. Hence, $g'$ is actually a cubic polynomial, even though $f$ is quadratic. We converted the constrained optimization problem to an unconstrained one, but in the process our quadratic optimization became a cubic optimization problem. Rosenberg~\cite{rosenberg1975reduction} showed that need not be the case if we restrict ourselves to binary variables over Boolean spaces, i.e. $x_1, x_2, y \in \bool$. One can use the following quadratic function:
\begin{align}
	\label{eq:rosenberg}
	h(x_1, x_2, y) = 3y + x_1x_2 -2x_1y -2 x_2y.
\end{align}
It is easy to verify that for $x_1, x_2, y \in \bool$,
\[
y = x_1 x_2 \Leftrightarrow h(x_1, x_2, y) = 0
\qquad \text{and} \qquad
y \neq x_1 x_2 \Leftrightarrow h(x_1, x_2, y) > 0.
\]
This shows, if our initial goal was to minimize a quadratic function $f(x_1, x_2, y)$ over Boolean variables $x_1, x_2, y \in \bool$, subject to the constraint $y = x_1 x_2$, we can minimize the quadratic function $g(x_1, x_2, y) = f(x_1,x_2, y)  + M \cdot h(x_1, x_2, y)$ without any constraint and use any Quadratic  Unconstrained Boolean Optimization hardware. 

\newcommand{\eqdef}{\stackrel{\rm def}{=}}


We next explore the question whether we can find an analogue of the Rosenberg polynomial in the Ising space. The answer turns out to be a little more complicated. We first show that such a polynomial does not exist when we allow just one extra variable. 

\begin{theorem}
    There is no quadratic polynomial $h(s_1, s_2, y)$, such that for all $s_1, s_2, y \in \ising$
    \[
    y = s_1 s_2 \Leftrightarrow h(s_1, s_2, y) = 0 
    \qquad \text{and} \qquad
    y \neq s_1 s_2 \Leftrightarrow h(s_1, s_2, y) > 0.
    \]
\end{theorem}

\begin{proof}
We want to construct a quadratic polynomial $ h $, such that (i) if $y = s_1 s_2$, then $h(s_1, s_2, y) = 0$; and (ii) if $y \neq s_1 s_2$, then $h(s_1, s_2, y) \geq 1$ (the choice of $1$ on right hand side is arbitrary; any other positive number will also work).
Let
\[
	h(s_1, s_2, y) = a_0 + a_1 s_1 + a_2 s_2 + a_3 y + a_{12} s_1 s_2 + a_{13} s_1 y + a_{23} s_2 y \eqdef s^{\top}a,
\]
where
\[
s = \begin{pmatrix}
	1 & s_1 & s_2 & y & s_1 s_2 & s_1 y & s_2 y
	\end{pmatrix}^{\top} \quad \text{and} \quad a = \begin{pmatrix}
		a_0 & a_1 & a_2 & a_3 & a_{12} & a_{13} & a_{23} 
	\end{pmatrix}^{\top}.
\]
Here $X^{\top}$ denotes transpose of a matrix $X$. The set of equality and inequality conditions can be captured by matrices $E$ and $F$ as follows.
\[
E = \begin{pmatrix}
		1 & -1 & -1 & 1 & 1 & -1 & -1
		\\
		1 & -1 & 1 & -1 & -1 & 1 & -1
		\\
		1 & 1 & -1 & -1 & -1 & -1 & 1
		\\
		1 & 1 & 1 & 1 & 1 & 1 & 1
	\end{pmatrix}
\text{ and }
F = \begin{pmatrix}
		1 & -1 & -1 & -1 & 1 & 1 & 1
		\\
		1 & -1 & 1 & 1 & -1 & -1 & 1
		\\
		1 & 1 & -1 & 1 & -1 & 1 & -1
		\\
		1 & 1 & 1 & -1 & 1 & -1 & -1
	\end{pmatrix}.
\]
Let $e \in \mathbb{R}^4$ denote the vector of all 1's: $e = \begin{pmatrix} 1 & 1 & 1 & 1 \end{pmatrix}^{\top}$. Let $\mathbf{0} \in \mathbb{R}^4$ denote the vector of all 0's. Then we have that
\[
Ea = \mathbf{0} \qquad \text{and} \qquad Fa \geq e.
\]
From $Ea = \mathbf{0}$, we get that $a$ is in the right kernel of $E$, which is given by the following matrix:
\[
	K \eqdef
	\begin{pmatrix} 
	    0 & 1 & 0 & 0 & 0 & 0 & -1 \\
	    0 & 0 & 1 & 0 & 0 & -1 & 0 \\
	    0 & 0 & 0 & 1 & -1 & 0 & 0
	\end{pmatrix}^{\top}.
\]
Let us write $a \eqdef Kb$ for some $b \in \mathbb{R}^3$. Since $Fa \geq e$, we have that $FKb \geq e$. Since
\[
FK = \begin{pmatrix}
		-2 & -2 & -2 \\
		-2 & 2 & 2 \\ 
		2 & -2 & 2 \\
		2 & 2 & -2
	\end{pmatrix},
	\quad \text{we have that} \quad
FKb = \begin{pmatrix}
	-2(b_1 + b_2 +b_3) \\
	-2(b_1 - b_2 - b_3) \\
	2(b_1 - b_2 + b_3) \\
	2(b_1 + b_2 - b_3)
	\end{pmatrix} 
	\geq e = \begin{pmatrix} 1 \\ 1 \\ 1 \\ 1 \end{pmatrix}.
\]
Observe that sum of entries of the vector $FKb$ is $0$, whereas for the inequality to hold, it should necessarily be at least $4$. Hence, no feasible solution exists.
\end{proof}

On the other hand, if we allow two variables then we come up with such a polynomial. To this end, we employ an extra variable $d$ and aim to find a {\it positive} quadratic polynomial $h(s_1, s_2, y, d)$ such that (i) when $y=s_1 s_2$, $h(s_1, s_2, y, d) = 0$ for some $d \in \{-1,1\}$; and (ii) when $y \ne s_1 s_2$,  $h(s_1, s_2, y, d) > 0$ for all $d \in \{-1,1\}$. By positive polynomial, we mean that $h(s_1, s_2, y, d) \geq 0$ for all choices of the four variables.


Observe that, for each choice of $ d $ as a function of $ s_1, s_2 $, the above constraints give rise to a linear system of inequalities, similar to the last section. If any of these 16 choices give us a feasible solution, then that solves our problem. We carry out some of these choices in an ILP solver and obtain the following solution:
\[
\fbox{
	$h (s_1, s_2, y, d) = 4 + s_1 + s_2 - y - 2d + s_1 s_2 - s_1 y - s_2 y - 2 s_1 d - 2 s_2 d + 2 y d$
}
\]

In Section \ref{sec:hobo-to-qubo} we  consider higher order unconstrained binary optimization problems. We show how introduction of new variables and repeated application of the above technique can help us reduce the problem to quadratic case. If we restrict ourselves only to Boolean variables,  \cite{dattani2019quadratization, anthony2017quadratic,rodriguez2018linear} are excellent survey of previous work in this domain. 

\section{Quadratization of higher order optimization problem}
\label{sec:hobo-to-qubo}

This section details out our quadratization of higher order Boolean polynomials. A formal definition of quadratization follows.

\begin{definition}\label{def:quadratization}
Let $\mathcal{B}$ denote either of the two sets: $\{0,1\}$ and $\{-1,1\}$. Given a higher order Boolean function $f:\mathcal{B}^n \rightarrow \mathbb{R}$, we say that any function $h(z,y)$ is a quadratization of $f(z)$ if $h(z,y)$ is a quadratic polynomial depending on $z$ and auxiliary variables $y = (y_1, \dots, y_m)$ such that
\[
f(z) = \min_{y \in \mathcal{B}^m} h(x,y), \quad \text{for all} \quad z \in \mathcal{B}^n.
\]
\end{definition}

From Definition~\ref{def:quadratization}, it is evidently clear that minimizing $f(z)$ is tantamount to minimizing $h(z,y)$ over $\mathcal{B}^{n+m}$. 
One way to do quadratization is to convert each monomial into a quadratic polynomial. In literature, family of such quadratization techniques are known as {\it termwise quadratization}. They are usually very effective when the domain of each variable is $\{0,1\}$ and the higher order polynomial is sparse. In fact, there are various techniques that reduces the number of auxiliary variable introduced when converting a monomial. For instance, if one wishes to minimize $f(x)$ and there is a negative monomial in $f(x)$, then one can use exactly one auxiliary variable and convert the monomial to a quadratic polynomial~\cite{freedman2005energy}. That is, the following two polynomials are equivalent:
$\prod_{i=1}^d x_i$ and $\min_{y \in \{0,1\}} \left\{(d-1)y - \sum_{i=1}^dx_iy\right\}$.
Likewise, Rodrigues-Heck~\cite{rodriguez2018linear} have shown that a positive monomial of degree $d$ can be replaced by $\lceil \log d \rceil - 1$ auxiliary variables. 

Unfortunately, these techniques have two main issues. First of all, the fact that a monomial over $\{0,1\}$ is highly likely to be $0$ are crucially important for such termwise quadratizations. When moving to $\{-1,1\}$ domain, these techniques do not work. Second of all, the termwise quadratization is by definition "local" in nature and does not cater to a "global" replacement of terms. We present two heuristic techniques that are independent of the domain space and focus on quadratization keeping the whole expression in mind. 

Before we discuss our algorithms for converting higher order Boolean functions to quadratic Boolean functions respecting Definition~\ref{def:quadratization}, we would like to stress on the fact that such a quadratization in the $\{-1,1\}$ domain space can potentially give significant savings in number of auxiliary variables. As an example, consider the following monomial
$\prod_{i=1}^n s_i$  where $s = (s_1, \dots, s_n) \in \{-1,1\}^n.$
From Section~\ref{sec:twoauxIsing}, it is evidently clear that the number of auxiliary variables required to convert this monomial to a quadratic polynomial over the domain space $\{-1,1\}$ is $O(n)$. If we try to convert the variables $x_i$ to variables over $\{0,1\}$, then the resulting polynomial will be
$\prod_{i=1}^n (1-2x_i)$ where $ x_i = \frac{1 - s_i}{2} \in \{0,1\}.$
This polynomial is dense and included all possible monomial terms in $x_i$. Even if we assume that each monomial can be quadratized by exactly one variable, the termwise quadratization will still require at least $\Omega\left(2^n\right)$ auxiliary variables. The algorithms presented below will also require significantly more variables in $\{0,1\}$ domain than $O(n)$ variables required when doing conversion in $\{-1,1\}$ domain space.

Having discussed the importance of quadratization in $\{-1,1\}$ domain space, we present our algorithms below. We stress that the final steps of both the algorithms are same. Post quadratization, we invoke either the Rosenberg polynomial (over Boolean space) or the polynomial described in Section~\ref{sec:twoauxIsing} (over Ising space) to impose the constraints between the auxiliary variables and the quadratic term they replace. Our Algorithm~\ref{fig:hobotoqubo1} is detailed out next.

\begin{algorithm}[htbp]
\caption{{\scshape Hobo To Qubo 1}}
\label{fig:hobotoqubo1}
\begin{algorithmic}[1]

\Require{A higher order binary optimization (HOBO) over $(z_1, z_2, \dots, z_n) \in \mathcal{B}^n$ where $\mathcal{B}$} is either $\{-1,1\}$ or $\{0,1\}$.

\Ensure{A QUBO equivalent to the given higher order binary optimization problem}

\State Sort the indices of variables and create a data structure of (key, value) pairs where key is the set of all quadratic terms appearing in HOBO and value is the set of all monomials of degree at least $3$.

\While{all the keys are deleted}

    \State Select the key with largest number of values and replace it with a variable

    \State Update the data structure by adding keys and values for new variable
    
    \State Delete all degree 3 terms that involved the key
    
    \State Delete the key if all the values of the key has been deleted

    \State Store the variable and the quadratic term it substitutes in a map
    
\EndWhile

\State Invoke the quadratic polynomial corresponding for each map of auxiliary variable and the quadratic term

\State {\bf Return} The QUBO equivalent to given HOBO problem.
\end{algorithmic}
\end{algorithm}

The idea of the algorithm is pretty simple. The algorithm starts with a hash table where keys are all possible pairs $\{(z_i, z_j) : 1 \leq i < j \leq n\}$ and values corresponding to a key $(z_i, z_j)$ are monomials of degree at least $3$ containing both $z_i$ and $z_j$. Our algorithm greedily replaces the key with largest number of values by an auxiliary variable and then update the hash table by introducing the auxiliary variable in the key. 
The degree of all values for which the key has been replaced by the auxiliary variable will decrease by 1. The values with degree $3$ containing the replaced key will become quadratic and they are deleted from the table. This is repeated until the table is empty. Finally, once the hash table is empty, we invoke either the Rosenberg polynomial (when $\mathcal{B} = \{0,1\}$) or the polynomial described in Section~\ref{sec:twoauxIsing} (when $\mathcal{B} = \{-1,1\}$) to impose the constraints between the auxiliary variables and the quadratic term they replace.

We now proceed to present our second heuristic as Algorithm~\ref{fig:hobotoqubo2}. This is again a greedy approach with a different objective to reduce. Simply put, we have a dynamic weight bipartite graph with quadratic terms on left hand side (LHS) and monomials of degree at least 3 on right hand side (RHS). There exists an edge between LHS and RHS vertices if and only if the RHS vertex contains the LHS vertex in the expression; the edge weight is simply one less the degree of the monomial. In this case, we proceed to replace the LHS vertex with maximum sum of edge weights with a new auxiliary variable and update the graph by introducing the quadratic terms involving the new variable in the graph. We remove all degree three monomials for the replaced LHS vertex (as they are quadratic now). We repeat these sequence of steps until all the RHS vertices become quadratic. 

\begin{algorithm}[htbp]
\caption{{\scshape Hobo To Qubo 2}}
\label{fig:hobotoqubo2}
\begin{algorithmic}[1]

\Require{A higher order binary optimization (HOBO) over $(z_1, z_2, \dots, z_n) \in \mathcal{B}^n$ where $\mathcal{B}$} is either $\{-1,1\}$ or $\{0,1\}$.

\Ensure{A QUBO equivalent to the given higher order binary optimization problem}

\State Sort the indices of variables and construct a weighted bipartite graph
    \State All possible combination of quadratic terms appearing in HOBO are left nodes and all monomials in the HOBO are right nodes
    
    \State Edges exist if monomial contains the quadratic term

    \State Edge weights are simply the degree of the monomial - 1

\While{there exists an edge on the graph}

    \State Replace the quadratic term with variable with largest sum of edge weights
    
    \State Remove all degree 3 terms that involved the quadratic term
    
    \State Remove quadratic terms if there is no edge originating from it
    
    \State Update the graph after adding quadratic terms involving new variable
    
    \State Store the variable and the quadratic term it substitutes in a map
    
\EndWhile

\State Invoke the quadratic polynomial corresponding for each map of auxiliary variable and the quadratic term

\State {\bf Return} The QUBO equivalent to given HOBO problem.
\end{algorithmic}
\end{algorithm}


\section{Applications}
Many real-world problems can be modeled as a HOBO. A large number of applications have been modeled in the quadratic case. For example,  in combinatorial scientific computing \cite{shaydulin2018community, shaydulin2019network, ushijima2017graph,ushijima2019multilevel},  chemistry \cite{hernandez2016novel, terry2019quantum}, and machine learning \cite{crawford2016reinforcement, khoshaman2018quantum, negre2019detecting} 
 In the following subsections, we give examples that consist of problems modeled as a HOBO with higher-order terms greater than two.
\subsection{Problems on Hypergraphs}
Many NP-hard problems on graphs can easily be encoded in the quadratic case as a QUBO \cite{lucas2014ising}. A hypergraph is a generalization of a graph such that a (hyper-) edge may contain more than two nodes. This generalization thus provides an area rich in problems that can be modeled as a HOBO.  For example, the Hypergraph Max-Covering problem \cite{he2013approximation}. Given a hypergraph $H=(V , E)$ with vertex set $V$ and hyperedge set $E$, with an associated edge weight $w(e) \in \mathbb{R}$, for each $e\in E$. The problem is to find a subset of nodes $V' \subset V$ such that the total weight of hyperedges completely covered by $V'$ is maximized. If $x_i \in \bool$ is 1 when $x_i \in V'$ and 0 otherwise, then this problem is simply 
\begin{equation}
    \max \sum_{e \in E} w(e)\prod_{i\in e}x_i.
\end{equation}
Another example is finding the min-vertex cover of a hypergraph where the problem is to determine a minimum carnality set $V' \subset V$ such that every hyperedge contains at least one node in $V'$. This is simply equivalent to solving
\begin{equation}
    \min \sum_{i=1}^n (1-x_i) + M\sum_{e \in E} \prod_{i\in e}x_i ,
\end{equation}
for any $M>n.$ In this case $x_i=0$ if $v_i \in V'$. These problems have applications to scheduling problems \cite{bansal2010inapproximability}. Within the Ising space, the Hypergraph Partitioning Problem and the Hypergraph MAX-CUT problem are natural examples. The hypergraph partitioning deals with partitioning the vertex set equally such that the number of cut-edges is minimized, while the Hypergraph MAX-CUT problem is to partition the vertex set such that the number of cut-edges is maximized, where a cut-edge is defined as a hyperedge that contains nodes in more than one part. If $s_i \in \ising$ represents the decision variable for node $i$ belonging to one part or the other, then we can show that the Hypergraph MAX-CUT problem is given by 
\begin{equation}
    \max \sum_{e \in E} \left( 1 - \frac{1}{4^{|e|}}\prod_{i \in e}(s_{i_e} + s_i)^2\right)
\end{equation}
where ${i_e}$ represents the smallest index of the nodes in $e$. Likewise the Hypergraph Partitioning Problem would minimize the same objective function subject to the balancing constraint $\sum_i s_i = 0$, which can be added as a quadratic penalty to give 
\begin{equation}
    \min \sum_{e \in E} \left( 1 - \frac{1}{4^{|e|}}\prod_{i \in e}(s_{i_e} + s_i)^2\right) + A(\sum_is_i)^2
\end{equation}
for a large enough $A>0$.

    

\subsection{(Weighted) Maximum satisfiability problem}
Maximum satisfiability (MAX-SAT) is a generalization of the well known Boolean satisfiability problem. Given a Boolean formula in conjunctive normal form, the goal is to find the assignment that maximizes the number of satisfying clauses. This problem is known to be NP-hard, because a MAX-SAT solver can be used for solving NP-complete SAT problems, where the goal is to decide whether there exists an assignment that satisfies all clauses or not. In weighted MAX-SAT problem, every clause is associated with a positive penalty score. The goal is to minimize the total penalty due to all non satisfying clauses. This can naturally be expressed as a higher order Boolean optimization problem as follows.

Consider a clause $c$, which is associated with a penalty  $p_c > 0$. Let, $S_c^+$ and $S_c^-$ be the sets of variables which are not negated and negated in $c$. E.g. if $c = \overline{x_0} \vee x_1 \vee x_2 \vee \overline{x_4}$, then $S_c^+ = \{x_1, x_2\}$ and $S_c^- = \{x_0, x_4\}$. If $C$ is the set of all clauses in a MAX-SAT problem, it is equivalent to minimization of the following polynomial over Boolean variables:
$ \sum_{c \in C} p_c \prod_{u \in S_c^+} (1-u) \prod_{v \in S_c^-} v$

\subsection{Electronic Structure calculations on adiabatic quantum annealers}
In recent works, Xia et al\cite{xia2017electronic} and Streif et al~\cite{streif2019solving} considered quantum chemistry problems of finding the exact ground state energy of various molecules using commercially available quantum annealers from D-Wave. D-Wave's quantum annealer is essentially a quadratic Ising solver and supports problem hamiltonian of the form $H_p = \sum_i h_i \sigma_z^i + \sum_{i,j} J_{ij} \sigma_z^i \sigma_z^j$. In their work, initially Xia et. al. reduced the molecular structure Hamiltonian problem to a '$k$-local' Ising Hamiltonian problem of the form $\sum_{i_1} h_{i_1} \sigma_z^{i_1} + \sum_{i_1,i_2} J_{{i_1}{i_2}} \sigma_z^{i_1} \sigma_z^{i_2} + \sum_{i_1,i_2, i_3} J_{{i_1}{i_2}{i_3}} \sigma_z^{i_1} \sigma_z^{i_2} \sigma_z^{i_3} + \cdots$. Afterwards, they converted the problem to Boolean domain and applied Rosenberg polynomial (Equation \eqref{eq:rosenberg}) iteratively to convert the problem to $2$-local Hamiltonian supported by D-Wave annealers. As discussed in Section \ref{sec:hobo-to-qubo}, this method introduces exponentially many auxiliary variables. In fact, \cite{xia2017electronic} mentions a single $\sigma_x^1 \sigma_x^2 \sigma_y^3 \sigma_y^4$ term in the original molecular Hamiltonian corresponds to more than $1000$ terms in the final $2$-local Ising Hamiltonian form. By performing the conversion directly in the Ising domain, we should be able to drastically reduce the number of auxiliary variables and terms.


\section{Experimental Results and Discussion}
A synthetic dataset of sparse instances of HOBO problems over Ising space was developed for experimentation purposes. The data set is available at \cite{isingdata}. Table \ref{tab:dataset} describes distribution of monomials in the dataset, after performing some basic preprocessing. For example, if a linear term has a relatively large positive coefficient, we know the corresponding variable must be $-1$. For converting the higher order problem to a quadratic one, we have four options: we can either convert it to Boolean space and use Rosenberg Polynomial to reduce its degree, or we can use our polynomial and directly perform the degree reduction in Ising space. In either of the option we can also choose to use either Algorithm 1 or Algorithm 2 for variable choice heuristic. Table \ref{tab:results} describes the number of variables and terms in the resultant quadratic polynomial. The performance of Algorithm 1 and Algorithm 2 are comparable to each other. However, as original problems are sparse in Ising domain, converting the problem using our polynomial results in significantly more compact representation compared to Boolean domain degree reduction using Rosenberg polynomial.

\begin{table}[htp]
    \centering
    \begin{tabular}{|*{16}{c|}} 
        \hline
        Dataset & Variables & \multicolumn{14}{c|}{Terms in Degree} \\
        \hline
        &&1&2&3&4&5&6&7&8&9&10&11&12&13&14\\
        \hline
        D20A & 15 & 15 & 105 & 60 & 53 & 49 & 49 & 48 & 37 & 20 & 23 & 12 & 4 & - & - \\
        \hline
        D20B & 14 & 14 & 91 & 60 & 55 & 38 & 31 & 10 & 5 & 6 & - & - & - & - & - \\
        \hline
        D20C & 15 & 15 & 105 & 62 & 47 & 52 & 33 & 46 & 49 & 26 & 22 & 26 & 17 & 7 & 1 \\
        \hline
        D30A & 17 & 17 & 136 & 98 & 61 & 50 & 30 & 28 & 22 & 23 & 6 & 3 & 1 & 2 & - \\
        \hline
        D30B & 18 & 18 & 153 & 130 & 66 & 50 & 41 & 35 & 14 & 12 & 4 & 2 & - & - & - \\
        \hline
        D30C & 20 & 20 & 190 & 114 & 65 & 58 & 50 & 44 & 24 & 23 & 7 & - & 2 & - & - \\
        \hline 
    \end{tabular}
    \caption{Distribution of number of monomials of different degrees in the dataset} \label{tab:dataset}
\end{table}

\begin{table}[ht]
    \centering
    \begin{tabular}{|c|cc|cc|cc|cc|} 
        \hline
        Dataset & \multicolumn{2}{c|}{Ising, Algo 1} & \multicolumn{2}{c|}{Ising, Algo 2} &  \multicolumn{2}{c|}{Boolean, Algo 1} & \multicolumn{2}{c|}{Boolean, Algo 2}\\
        \hline
        &variables&terms&variables&terms&variables&terms&variables&terms \\
        \hline
        D20A  &  561 & 2581&  585 & 2685 & 597 & 26025 & 583 & 25983\\
        \hline
        D20B & 274 & 1290 &  278 & 1307 &  303  & 4273 & 314 & 4306\\
        \hline
        D20C & 621 & 2857 &  629 &2888 &  730 & 33429 & 1035 & 34344 \\
        \hline
        D30A & 545 & 2493 & 541 & 2489 &  1034 & 31189 & 1003 & 31096 \\
        \hline
        D30B & 512 & 2405 &  526 & 2466 &  751 & 15397 &  800 & 15544 \\
        \hline
        D30C & 706 & 3230 & 716 & 3269 &  1478 & 28103 &  1414 & 27911\\
        \hline

    \end{tabular}
    \caption{Experimental Results} \label{tab:results}
\end{table}
 Our experimental results demonstrate that the proposed approach facilitates compressed quadratization of higher order Ising problems.
Future work includes integrating this methodology into real-world applications. 
\newpage
\Section{References}


\end{document}